\documentclass[12pt]{article}
\usepackage{amsmath,amssymb,amsthm,mathrsfs}
\usepackage{latexsym,graphicx,epsfig}

\newtheorem{theorem}{Theorem}[section]
\newtheorem{Def}{Definition}[section]
\newtheorem{prop}{Proposition}[section]
\newtheorem{lemma}{Lemma}[section]
\newtheorem{corol}{Corollary}[section]
\newtheorem{remark}{Remark}[section]
\newtheorem{hyp}{Hypothesis}[section]
\newtheorem{exap}{Example}[section]

\newcommand{\E}{\mathbb{E}}
\newcommand{\pr}{\mathbb{P}}
\newcommand{\R}{\mathbb{R}}

\newcommand{\ra}{\rangle}
\newcommand{\la}{\langle}

\newcommand{\expect}[1]{\mathbb E\left[#1\right]}
\newcommand{\integral}[3]{\int_{#1}^{#2}{#3}\;}

\title{Hedging of defaultable claims in a structural model using a locally risk-minimizing approach}

\author{Ramin Okhrati\thanks{
University of Southampton, Southampton, UK,
Email: r.okhrati@soton.ac.uk. 
The author gratefully acknowledges the partial financial support of WWTF grant MA09-005 at Vienna University of
 Technology, Austria. }~, Alejandro Balb\'as\thanks{Universidad Carlos III de Madrid, Madrid, Spain, Email: balbas@emp.uc3m.es. The author thanks the partial financial support of ``Comunidad Aut\'onoma de Madrid" (Spain), Grant S2009/ESP-1594, and ``Ministerio de Econom\'{\i}a" (Spain), Grants ECO2009-14457-C04 and ECO2012-39031-C02-01.}~ and Jos\'e Garrido\thanks{Concordia University, Montreal, Canada, Email: jose.garrido@concordia.ca. The author gratefully acknowledges the financial support of NSERC grant 36860-2012.}
}

\date{}
\makeatletter

\newcommand{\Rmnum}[1]{\expandafter\@slowromancap\romannumeral #1@}
\makeatother

\usepackage{fullpage}
\begin{document}

\maketitle

\begin{abstract}
In the context of a locally risk-minimizing approach, the problem of hedging defaultable claims and their F\"ollmer-Schweizer decompositions are discussed in  a structural model. This is done when the underlying process is a finite variation L\'evy process and the claims pay a predetermined payout at maturity, contingent on no prior default. More precisely, in this particular framework, the locally risk-minimizing approach is carried out when the underlying process has jumps, the derivative is linked to a default event, and the probability measure is not necessarily risk-neutral. 

\end{abstract}
F\"ollmer-Schweizer decomposition, Galtchouk-Kunita-Watanabe decomposition
\setlength{\baselineskip}{1.2\baselineskip}

\pagenumbering{arabic}

\newpage
\section{Introduction}\label{chap:icr}
\addcontentsline{toc}{chapter}{Introduction to Credit Risk}
\setcounter{equation}{0}

In its simple form, a defaultable claim pays a certain pre-defined amount at the maturity of the contract, if there has not been a prior default, and pays zero otherwise. In this work, an hedging analysis is carried out for these derivatives when the underlying risky asset is modeled by a  finite variation L\'evy process. It is of mathematical and practical interest
to study the hedging of defaultable claims when the asset prices are affected by jumps. The extension to more complicated derivatives and underlying processes will be interesting for future work. First, we review the literature and related previous works.


We start by a definition of credit risk. Credit risk is the risk
associated with the possible financial losses of a derivative caused
by unexpected changes in the credit quality of the counterparty's
issuer to meet its obligations. The first paper that introduced credit risk for a path independent claim goes back to the
work of 
Merton (1974). 

When analyzing a credit derivative, normally there are two prominent issues, pricing and hedging of the derivative. The latter is a more challenging question, especially when the market is incomplete. In most financial models, even when working with simple stochastic processes, a complete hedge still
may not be feasible for credit derivatives. There are different approaches to manage the risk in an incomplete market. Quadratic hedging is a well developed and applicable method to manage the risk.


Schweizer (2001) or Pham (1999) provide a good survey of quadratic hedging methods in
incomplete markets. In Schweizer (2001) two quadratic hedging approaches are discussed for the case where the firm's value process
is a semimartingale. These are local risk-minimization and mean-variance hedging.

If we prefer a self-financing portfolio in order to hedge a contingent claim, we
speak of mean-variance hedging. If we rather select a portfolio with the same terminal value as the contingent claim (but not necessarily self-financing), we are in the context of a (locally) risk-minimizing approach. Schweizer, Heath and Platen (2001) provide a comprehensive study and comparison of both approaches. In our paper a local risk-minimization approach is used to manage the risk associated with the defaultable claims.

Local risk-minimization hedging emerged in the development of the concept of risk minimization.
F\"ollmer
and Sondermann (1986) were among the first to deal with this problem. They solved the problem identifying the risk-minimization strategy when the underlying process is a martingale. The generalization to the local martingale case is done in Schweizer (2001). The solution of the risk-minimization problem is linked to the so called Galtchouk-Kunita-Watanabe (GKW) decomposition assuming that the underlying process is a local martingale.

For a non-martingale process, Schweizer (1988) provides an example
of an attainable claim that does not admit a risk-minimization strategy.
The extension is possible by putting more restrictive conditions on the
underlying process as well as on the hedging strategies.

Literally saying, one has to pay more attention to the local properties of the problem. As for the role of the underlying process, it has to satisfy the structure condition\footnote{Assume that $X$ is a square integrable special semimartingale with the canonical decomposition $X=X_0+M+A$. Then $X$ satisfies the structure condition, if there exists a predictable process $\lambda$ such that $A_t=\int_0^t\lambda_sd\la M\ra_s$ for all $0\leq t\leq T$, and the mean-variance tradeoff process, defined by $K_t=\int_0^t\lambda_s^2d\la M\ra_s$, is $\pr$-almost surely finite for all $0\leq t\leq T$.} (SC), see Schweizer (1991) or Schweizer
(2001). Under certain conditions like SC, a locally risk-minimizing strategy
is equivalent to a more tractable one, called pseudo-locally risk-minimizing strategy. F\"ollmer and
Schweizer (1991) gives a necessary and sufficient condition for the
existence of a pseudo-locally risk-minimizing strategy. It turns out that finding these strategies is equivalent to the existence of a generalized version of the GKW decomposition, known as the F\"ollmer-Schweizer (FS)
decomposition. A sufficient condition for the existence of an FS  decomposition
is provided by Monat and Stricker (1995).

Although the existence of locally risk-minimizing strategies is proved under some conditions, it completely depends on the FS decomposition. In some special cases there are constructive ways of finding this decomposition
explicitly. 
The case of continuous processes is more flexible, and the well known method of minimal equivalent
local martingale measure (MELMM) is applicable. 
Biagini and Cretarola (2009) study general defaultable markets under a locally risk-minimizing approach. However the continuity of the underlying process is a crucial assumption in their work.


Most recently,
Choulli, Vandaele and Vanmaele (2010) find an explicit form of the FS decomposition based on a representation theorem. 
They aimed to provide a general framework under which the FS decomposition is obtained. While this work could fit into theirs, our approach leads to a more explicit form of the FS decomposition. By using a slightly different method, we specifically focus on the hedging of the defaultable claims, based on the theory of local risk-minimization, assuming that the underlying
process is a bounded variation L\'evy process with positive drift.

Our paper studies a structural model, in the sense that a default event is defined, and we use the whole market information represented by the filtration generated by the underlying process. However, while the default event is structural (and so economically intuitive), we use an analysis like that of reduced form models and especially intensity based models. These models were pioneered by the works
of Artzner and Delbaen (1995) or Jarrow and Turnbull (1995), and they do not use or
determine a default model of the firm. They use an intensity process or hazard process instead.

Martingale techniques and the idea of intensity in reduced form models are applied
to analyze the structure of the defaultable claims. 
In Section \ref{sec:mp}, under some conditions, a compensation formula is used to find a canonical decomposition of the defaultable process $Z=\left(f(t,X_t)1_{\{\tau>t\}}\right)_{t\geq0}$, where $\tau$ is a hitting time (defining the default time) and $f=f(t,x)$ is a real valued function. This enables us to use compensator techniques for these types of processes. 

The predictable finite variation part of this decomposition is absolutely continuous with respect to the Lebesgue measure. Hence, when $f$ is a constant function, this precisely determines the intensity of $\tau$. This intensity is already obtained  by Theorem 1.3 of Guo and Zeng (2008) for a Hunt process that has finitely many jumps on every bounded interval. However, a finite variation L\'evy process could have infinitely many jumps on bounded intervals and so some modifications of their ideas would be essential.

Note that in our analysis the underlying process allows for jumps, the payoff is linked to a structural default event, and the probability measure is not necessarily a martingale measure. In addition we do not use any type of Girsanov's theorem, but the results are based on solutions of partial integro-differential equations (PIDE). We also study the structure of the default indicator process $\left(1_{\{\tau>t\}}\right)_{t\geq0}$ and finite horizon ruin time. Apart from the theoretical concerns in this paper, the main effort is devoted to obtain answers to two interesting questions.

The first question is, given a defaultable claim, how
a locally risk-minimizing hedging strategy can be carried out. As it is not possible to eliminate the credit risk completely,
the second question is whether it is possible to design a
customized defaultable security, to make the product
completely hedgeable i.e., the claim can be written as the sum of a constant and a stochastic integral with respect to the underlying process. This will result in a risk-free defaultable claim.  In our setup we find necessary and sufficient conditions for the existence of
such a product.


The paper is structured as follows. The model, some preliminary assumptions, and results are provided in Section \ref{sec:mp}. A canonical decomposition of the stochastic process $\left(f(t,X_t)1_{\{\tau>t\}}\right)_{t\geq0}$ is discussed in Section \ref{sec:drg}. This is an essential tool in our analysis. Locally risk-minimizing hedging strategies for defaultable claims are obtained in Section \ref{sec:hsdc}. In Section \ref{sec:edtpt}, we take a look at the structure of the default time.



\section{The Model and Preliminaries}\label{sec:mp}

We study a process $X$, modeling a firm's assets value, constructed on a probability space $(\Omega,\mathfrak F, \pr)$ and we denote by $\mathfrak F^X$
its natural filtration, completed and regularized so that it satisfies the usual conditions. 

We study defaultable claims with actual payoffs of the form
\begin{equation}\label{eq:payoff}
     F(X_T)1_{\{\tau>T\}},
\end{equation}
where $X_0>0$,   $F:\mathbb{R}\rightarrow\mathbb{R}$ is a real valued function, and $T>0$ is the maturity or expiration time of the
security, and
\begin{equation}\label{eq:default}
 \tau=\inf\{t;X_t<0\}.
\end{equation}
Note that the firm's assets value is assumed to be observable. Therefore from a financial point of view, the definition in \eqref{eq:default} for default makes sense if either
the modeler is the firm's management, the accounting data are publicly available or they can be well estimated in the market.

The security in \eqref{eq:payoff} pays $F(X_T)$ if there is no default in $[0,T]$ and
zero otherwise, hence the recovery rate is considered to be zero. 
A defaultable zero-coupon bond is a special case of this security by letting
$F(x)=c$, on $\mathbb{R}$ for a constant $c>0$. 
Later we will see that the function $F=F(x)$ is the boundary condition of a PIDE.



For two semimartingales $X$ and
$Y$, the notations $[X,Y]$ and $\la X,Y\ra$, respectively, stand for
quadratic covariation and conditional quadratic covariation, see Section 6, Chapter \Rmnum{2} and  Section 5, Chapter \Rmnum{3} of Protter (2004) or Section 4, Chapter \Rmnum{1} of Jacod and Shiryaev (1987) for the definitions. For the sake of completeness, we recall some basic
definitions.

The set of all uniformly integrable martingales is denoted by $\mathcal M$, and $\mathcal M^2$ is the set of all square integrable martingales, i.e. the set
of all martingales $X$ such that  $\sup_{t\geq0}\mathbb E[X_t^2]<\infty$. If in addition $X_0=0$, the notation $\mathcal M^2_0$ is used. Also the set  of integrable variation processes (starting at zero) is represented by $\mathscr A$.
In what follows, if $\mathscr{C}$ is a class of processes, its localized class is denoted by $\mathscr C_{loc}.$

One of the fundamental results in the theory of stochastic calculus
is the following, see Proposition 4.50, Chapter \Rmnum{1} of Jacod and Shiryaev (1987) for the proof.
\begin{prop}\label{prop:[X,X]=0}
   Assume that the process $X=(X_t)_{t\geq0}$ is a local martingale (i.e. $X\in\mathcal M_{loc}$). Then $X=X_0$ almost surely if and only if $[X,X]=0.$
\end{prop}
\begin{Def}
   The two processes $X$ and $Y$ belonging to $\mathcal M_{loc}$ are called orthogonal to each other if $XY$ belongs to $\mathcal M_{loc}$.
\end{Def}
If the processes $X$ and $Y$ belong to $\mathcal M^2_{loc}$, then it can be proved that $X$ is orthogonal to $Y$ if and only if $\la X,Y\ra=0$, for example see Theorem 4.2, Chapter \Rmnum{1} of  Jacod and Shiryaev (1987). If $X$ and $Y$ are two local martingales and $[X,Y]$ belongs to $\mathscr A_{loc}$, then by using the fact that $\la X, Y\ra$ is the compensator of $[X,Y]$, one can still show that $X$ is orthogonal to $Y$ if and only if $\la X,Y\ra=0$. 

A similar result to Proposition \ref{prop:[X,X]=0} is still true for the conditional quadratic variations as
well, and it is in fact a result that we use later.
\begin{corol}\label{corol: <x,x>=0}
   Suppose that $X$ belongs to $\mathcal M^2_{loc}$ or $[X]\in\mathscr A_{loc}$  then $X=X_0$ almost surely if and only if $\la X,X\ra=0$.
\end{corol}
\begin{proof}
Note that if $X$ is in $\mathcal M^2_{loc}$ then $[X]\in \mathscr A_{loc}$, see Proposition 4.50, Chapter \Rmnum{1} of Jacod and Shiryaev (1987). 
So it is enough to prove the result for $[X]\in\mathscr A_{loc}$. In this case, the result follows from Proposition \ref{prop:[X,X]=0} and the fact that 
$\la X\ra$ is the compensator of $[X]$. 
\end{proof}
Now, we explain our model assumptions.

This work is motivated by the first basic question of how the riskiness of a corporate bond can be managed. Such bonds represent a special defaultable claim for the firm. Specifically we focus on finite variation L\'evy processes modeling the firm's assets value. See Geman (2002) for some motivations on how these processes model the dynamic of stock prices better than diffusion or jump-diffusion models. Besides, some technical reasons also motivate this choice.

The following hypothesis is used throughout the paper and especially in Section \ref{sec:drg} to find the canonical decomposition of the process $\left(f(t,X_t)1_{\{\tau>t\}}\right)_{t\geq0}$, where $f=f(t,x)$ is a $C^{1,1}$ real-valued function.

\begin{hyp}\label{hyp:x-integrability1}
It is assumed that the firm's assets value process $X$, starting at $X_0=u>0$, is a bounded variation L\'evy process with L\'evy triplet $(\gamma,0,v)$, where the L\'evy measure $v$ is concentrated on $\mathbb R-\{0\}$. The process $X$ has the following L\'evy-It\^o decomposition
\begin{equation}\label{eq:bounded variation levy process}
   X_t =u+\mu t +\int_0^t\int_{\mathbb R}x\;J_X(ds\times dx),\quad t\geq 0,
\end{equation}
where $\mu=\gamma-\int_{[-1,1]}x\;v(dx)$ and $J_X$ is the jump measure of the process $X$. It is supposed that $\mu>0$ and $0<\int_{\mathbb R} x^2\;v(dx)<\infty.$ Also in the case of $v(\R)<\infty$, we assume that the measure $v$ is continuous. 
\end{hyp}
Note that the process $X$ in Hypothesis \ref{hyp:x-integrability1} has either finite activity\footnote{In this case the process $X$ is nothing but a compound Poisson process plus drift starting at $u>0$.} ($v(\R)<\infty$) or infinite activity ($v(\R)=\infty$). In the former case, by Remark 27.3 of Sato (1999), the compound Poisson process part of \eqref{eq:bounded variation levy process}  has a continuous distribution on $\R-\{0\}$, and in the latter one by Theorem 27.4 of Sato (1999), $X_t$ has a continuous distribution for every $t>0$. In particular, in either case we have $\pr(X_T=0)=0$. Hence, one can assume that the domain of the function $F$ in \eqref{eq:payoff} is the positive real line. This is because of the fact that $F(X_T)1_{\{\tau>T\}}=F_{|(0,\infty)}(X_T)1_{\{\tau>T\}}$, almost surely, where $F_{|(0,\infty)}$ is the restriction of $F$ to $(0,\infty)$.

By Remark \ref{remark:tist}, the default time $\tau$ given by \eqref{eq:default} is a totally inaccessible stopping time. Total inaccessibility of $\tau$ guarantees the unpredictability of default. 

\begin{remark}
Regarding a financial risk process, for instance a stock price process, the process $\left(e^{X_t}\right)_{t\geq0}$ and a non-zero barrier level are preferred from a financial point of view. However, this case can be covered by our model. For example, suppose that default is defined by $\tau=\inf\{t;e^{X_t}< c\}$, for the constant barrier $0<c<e^u$.  This is equivalent to $\tau=\inf\{t;X_t< \log c\}$, and regarding \eqref{eq:default}, there is nothing special about crossing the level zero here, this is just for ease of notation.
\end{remark}
The following hypothesis and definition are used in Sections \ref{sec:hsdc} and \ref{sec:edtpt} in order to find locally risk-minimizing hedging strategies and the distribution of the default time.


\begin{hyp}\label{hyp:integrability}
     Given a $C^{1,1}$ function $f=f(t,x)$ and a subinterval $\mathcal O$ of $[0,\infty)$, we say that it satisfies the integrability condition if the following holds for all $t$ in $\mathcal O$:
$$\int_{\mathbb R}\left| f(t,x+y)-f(t,x)\right|\;v(dy)<\infty,\quad \text{for all $x$ in $\R$}.$$
\end{hyp}

\begin{Def}\label{assumption:*}
   A function $F=F(x)$ belongs to class (*) if there is a $C^{1,1}$ function $f=f(t,x)$ that is the solution of the following PIDE
   \begin{equation*}
      \mathfrak{A}f(t,x)=\frac{\Big(\mathfrak A K(t,x)-x\mathfrak A f(t,x)-\beta f(t,x)\Big)}{\int_{\mathbb R}y^2\;v(dy)}\beta,\quad\text{for all $0\leq t\leq T$, and $x>0,$}
   \end{equation*}
   and $$f(T,x)=F(x),\;\text{for all real numbers}\;x>0,$$
   where $K(t,x)=xf(t,x)$, $\beta=\mu+\int_{\mathbb R} y\;v(dy)$, and the operator $\mathfrak A$ is given by
   \begin{equation}\label{eq:compensator-second approach-operator}
      \begin{split}
          \mathfrak A f(t,x)  ={}&  \frac{\partial f}{\partial t}(t,x)+\mu\frac{\partial f}{\partial x}(t,x)-\int_{(-\infty,-x]}f(t,x+y)\;v(dy)\\
             & +\int_{\mathbb R}\Big(f(t,x+y)-f(t,x)\Big)\;v(dy),\quad t\geq0,\quad x>0.
      \end{split}
   \end{equation}
It is also assumed that the functions $f$ and $K$ satisfy the integrability condition of Hypothesis \ref{hyp:integrability} on the interval $\mathcal O=[0,T]$.
\end{Def}
\begin{remark}
First note that given the integrability condition of Hypothesis \eqref{hyp:integrability}, the expression \eqref{eq:compensator-second approach-operator} is well defined in the sense that the integrals are finite. The PIDE in Definition \ref{assumption:*} will help to obtain strategies for the case when $X$ is not a martingale. This assumption can be thought of as a substitution for the change of probability measure.

In general, the existence of a classical solution for the PIDE in Definition \ref{assumption:*} is not always guaranteed. However, if $\beta=0$ (i.e. when $X$ is a martingale), then under some regularity conditions a classical solution can be provided by Feynman-Kac's representations. For a full discussion, examples, and many useful references, we refer the reader to Chapter 12 of Cont and Tankov (2004). In short, the main problem is that since we are in a pure
jump model, there is no diffusion and hence the proposed Feynman-Kac representation is not necessarily $C^{1,1}$. In cases where this smoothness holds then the Feynman-Kac representation is in fact a solution. Examples 1 and 2 of Cont, Tankov and Voltchkova (2004) show how the regularity can be easily  violated. If the smoothness does not hold or when $\beta$ is non-zero, then some approximation techniques must be used; in practice, viscosity solutions can be applied. Extending the results to the non-smooth case is left for future work.
\end{remark}

Finally, it is supposed that the market is frictionless and made of only two assets, a risky asset modeled by a process satisfying Hypothesis \ref{hyp:x-integrability1}, and a risk-free one. For simplicity, it is supposed that the value of the risk-free asset is equal to 1 at all times, i.e. the interest rate is zero.


\section{The Canonical Decomposition of $\left(f(t,X_t)1_{\{\tau>t\}}\right)_{t\geq0}$}\label{sec:drg}

In this section we investigate the canonical decomposition of the process $Z=\left(Z_t\right)_{t\geq0}$, where $Z_t=f(t,X_t)1_{\{\tau>t\}}$ and $f:[0,\infty)\times\mathbb R\rightarrow\mathbb R$ is a $C^{1,1}$ function.  More precisely, under some conditions we prove that it is a special semimartingale and we find a closed form for its finite variation predictable part. This result is used in Section \ref{sec:hsdc}.


\begin{theorem}\label{theorem:G-compensator-second approach}
   Assume that $X$ satisfies Hypothesis \ref{hyp:x-integrability1}. Let $f:[0,\infty)\times\mathbb R\rightarrow\mathbb R$ be a $C^{1,1}$ function that
satisfies the integrability condition of Hypothesis \ref{hyp:integrability} on $\mathcal O=[0,\infty)$. Then the process $Z$, where $Z_t=f(t,X_t)1_{\{\tau>t\}}$, is a special semimartingale and the process
   \begin{equation}\label{eq:compensator-second approach}
      \left(Z_t-Z_0-\int_0^t\mathfrak A f(s,X_{s})1_{\{\tau > s\}}\;ds\right)_{t\geq0},
       \end{equation}
   is an $\mathfrak F^X$- local martingale, where the stopping time $\tau$ is defined by
\begin{equation}\label{eq:default-2}
 \tau=\inf\{t;X_t<0\},
\end{equation}
and the operator $\mathfrak A$ is given by \eqref{eq:compensator-second approach-operator}. 
\end{theorem}
\begin{proof}
   Because the function $f$ is a $C^{1,1}$ function, the process $\left(f(t,X_t)\right)_{t\geq0}$ is a semimartingale and so by using the product formula of semimartingales, for $t\geq0$ we have
   \begin{equation}\label{eq:compensator-second approach-product formula}
f(t,X_t)1_{\{\tau\leq t\}} = \int_0^t1_{\{\tau< s\}}\;df(s,X_s)
        +\int_0^tf(s,X_{s^-})\;d1_{\{\tau\leq s\}}
        +[f(.,X_.),1_{\{\tau\leq .\}}]_t.
   \end{equation}
   To get the canonical decomposition of $\widetilde{Z}=\left(f(t,X_t)1_{\{\tau\leq t\}}\right)_{t\geq0}$, we prove that the processes defined by each of the three terms on the right-hand side of the above equation are special semimartingales and obtain their canonical decomposition. The rest of the proof is divided into four steps.

\textbf{Step 1.}   Since $f$ is a $C^{1,1}$ function, by applying It\^o's formula, we have that
   \begin{align*}
      f(t,X_t) = {} & f(0,X_0)+\int_0^t\frac{\partial f}{\partial s}(s,X_s)\;ds+\mu\int_0^t\frac{\partial f}{\partial x}(s,X_s)\;ds \\
            & + \int_0^t\int_{\mathbb R}\Big(f(s,X_{s^-}+y)-f(s,X_{s^-})\Big)\;J_X(ds\times dy),
   \end{align*}
see Theorem 4.2 of Kyprianou (2006) for a proof. By the compensation formula, see Theorem 4.4 of Kyprianou (2006), we get that
   \begin{eqnarray*}
      &&\E[\int_0^t\int_{\mathbb R}H_s\Big(f(s,X_{s^-}+y)-f(s,X_{s^-})\Big)\;J_X(ds\times dy)]=\\ &&\E[\int_0^t\int_{\mathbb R}H_s\Big(f(s,X_{s^-}+y)-f(s,X_{s^-})\Big)\;v(dy)ds],
   \end{eqnarray*}
   for all bounded non-negative predictable processes $H$, with the understanding that one of the expectations is well defined if and only if the other one is well defined as well and they are equal.
   Hence by the integrability condition of  Hypothesis \ref{hyp:integrability} and using Corollary 4.5 of  Kyprianou (2006), one can show that $f(t,X_t)=f(0,X_0)+M_t+\Lambda^f_t,$ for $t\geq0$, where $M$ is an $\mathfrak F^X$- local martingale and $\Lambda^f$ is a predictable finite
   variation process. The process $\Lambda^f$ is given by $\Lambda^f_t=\int_0^t\mathcal A f(s,X_s)\;ds$, where the operator $\mathcal A$ is defined by
   $$\mathcal A f(s,x)=\frac{\partial f}{\partial s}(s,x)+\mu\frac{\partial f}{\partial x}(s,x)+\int_{\mathbb R}(f(s,x+y)-f(s,x))\;v(dy),\quad s\geq0,\quad x\in\mathbb R.$$
  This proves that $\left(f(t,X_t)\right)_{t\geq0}$ and hence $Z$ are special semimartingales. Therefore $$\int_0^t1_{\{\tau< s\}}\;df(s,X_s)=\int_0^t1_{\{\tau<s\}}\;dM_s+\int_0^t1_{\{\tau< s\}}\mathcal A f(s,X_{s})\;ds.$$ 
Since the first term on the right-hand side of the above is a local martingale,  the first term of \eqref{eq:compensator-second approach-product formula} is a special semimartingale and its predictable finite variation part is then given by
   \begin{equation}\label{eq:compensator-second approach-compensator1}
      \left(\int_0^t1_{\{\tau< s\}}\mathcal A f(s,X_{s})\;ds\right)_{t\geq0}.
   \end{equation}

\textbf{Step 2.} Since $\left(1_{\{\tau \leq s\}}\right)_{s\geq0}$  is a special semimartingale, the second term of \eqref{eq:compensator-second approach-product formula} is also a special semimartingale. To find its canonical decomposition, we consider two cases. First $v(\R)<\infty$, in this case, the process $X$ is a compound Poisson process plus drift starting at $u>0$, and there are finitely many jumps on every bounded interval. Furthermore, since $\mu>0$, one can easily check that for every $x\geq0$, $\pr_x(\tau=0)=0$. Also by Remark \ref{remark:tist}, the stopping time $\tau$ is now totally inaccessible, hence Theorem 1.3 of Guo and Zeng (2008) is applicable, and so the process $$ \left(\int_0^tf(s,X_{s^-})\;d1_{\{\tau\leq s\}}-\int_0^t\int_{(-\infty,-X_s]}f(s,X_{s})1_{\{\tau>{s}\}}\;v(dy)\;ds\right)_{t\geq0}$$ is an $\mathfrak F^X$- local martingale. 

Now assume that $v(\R)=\infty$, then by Theorem 21.3 of Sato (1999) there are infinitely many number of jumps on every bounded interval.  Therefore the result of Guo and Zeng (2008) is not directly applicable this time.  However, since $\mu>0$, every $x\geq0$ is an irregular point for $(-\infty,0]$, which implies that $\pr_x(\tau=0)=0$, see Theorem 6.5 of Kyprianou (2006) and the discussions following it. So, using the compensation formula, their proof actually shows that $$\expect{\int_0^\infty H_s\;d1_{\{\tau\leq s\}}}=\expect{\int_0^\infty H_s1_{\{\tau>s\}}\int_{(-\infty,-X_s]}\;v(dy)\;ds},$$
with the understanding that the left-hand side is well defined if and only if the right-hand side is well defined as well and they are equal. Also by Lemma \ref{lem:finitevalue}, for every $t>0$, $\int_0^t\int_{(-\infty,-X_s]}1_{\{\tau>s\}}\;v(dy)\;ds$ is almost surely finite, and so by Lemmas 3.10 and 3.11 of  Chapter \Rmnum{1} of Jacod and Shiryaev (1987), the process $\left(\int_0^t\int_{(-\infty,-X_s]}1_{\{\tau>s\}}\;v(dy)\;ds\right)_{t\geq0}$ belongs to $\mathscr A_{loc}$.

Therefore, by Lemma \ref{lem:compensator-second approach}, in either case, the predictable finite variation part of the second term in \eqref{eq:compensator-second approach-product formula} is given by
   \begin{equation}\label{eq:compensator-second approach-compensator2}
      \left(\int_0^t\int_{(-\infty,-X_s]}f(s,X_{s})1_{\{\tau>{s}\}}\;v(dy)\;ds\right)_{t\geq0}.
   \end{equation}

 \textbf{Step 3.} Finally we find the canonical decomposition of the third term in \eqref{eq:compensator-second approach-product formula}.
   The indicator process $\left(1_{\{\tau\leq t\}}\right)_{t\geq0}$ is of finite variation. Then by Proposition 4.49(a), Chapter \Rmnum{1} of Jacod and Shiryaev (1987), we obtain
   \begin{equation*}\label{eq:compensator-second approach-1}
      [f(.,X_.),1_{\{\tau\leq .\}}]_t=\int_0^t\Delta f(s,X_s)\;d1_{\{\tau\leq s\}}.
   \end{equation*}
  This is a special semimartingale, and one can show that it belongs to $\mathscr A_{loc}$. Therefore, by Lemma \ref{lem:compensator-second approach}, to obtain the predictable finite variation part of the process, we need to calculate the following expectation
   $$\E\left[\int_0^{\infty}H_s\;d[f(.,X_.),1_{\{\tau\leq .\}}]_s\right]=\E\left[\int_0^{\infty}H_s\Delta f(s,X_s)\;d1_{\{\tau\leq s\}}\right],$$
   for an arbitrary bounded non-negative predictable process $H$. Again, the calculations of this expectation follow almost the same lines as those of Theorem 1.3 in Guo and Zeng (2008) and similar to the second case of Step 2, where the compensation formula is used.
   From there we obtain that the expectation $\E\left[\int_0^{\infty}H_s\;d[f(.,X_.),1_{\{\tau\leq .\}}]_s\right]$ is equal to $$\E\left[\int_0^{\infty}H_s1_{\{\tau>s\}}\int_{(-\infty,0]}(f(s,y)-f(s,X_{s}))\;v(dy-X_{s})\;ds\right],$$
with the understanding that one of the expectations is well defined if and only if the other one is well defined as well and they are equal.
Because of the integrability condition of Hypothesis \ref{hyp:integrability}, the assumptions of Lemma \ref{lem:compensator-second approach} are in force, and therefore the process $$ \left([f(.,X_.),1_{\{\tau\leq .\}}]_t-\int_0^t\int_{(-\infty,0]}(f(s,y)-f(s,X_{s}))1_{\{\tau>s\}}\;v(dy-X_{s})\;ds\right)_{t\geq0}$$ is an $\mathfrak F^X$- local martingale. Hence the predictable finite variation part of the third term in \eqref{eq:compensator-second approach-product formula} is given by
   \begin{equation}\label{eq:compensator-second approach-compensator3}
      \left(\int_0^t\int_{(-\infty,0]}(f(s,y)-f(s,X_{s}))1_{\{\tau>s\}}\;v(dy-X_{s})\;ds\right)_{t\geq0}.
   \end{equation}

   \textbf{Step 4.} From equations \eqref{eq:compensator-second approach-compensator1}, \eqref{eq:compensator-second approach-compensator2}, and \eqref{eq:compensator-second approach-compensator3}, we conclude that the predictable finite variation part of the process $$\Big(f(t,X_t)1_{\{\tau\leq t\}}\Big)_{t\geq0}$$ is equal to
   \begin{align*}
      \Big({} & \int_0^t1_{\{\tau< s\}}\mathcal A f(s,X_{s})\;ds+\int_0^t\int_{(-\infty,-X_s]}1_{\{\tau> s\}}f(s,X_{s})\;v(dy)\;ds \\
      & + \int_0^t\int_{(-\infty,0]}(f(s,y)-f(s,X_{s}))1_{\{\tau> s\}}\;v(dy-X_{s})\;ds\Big)_{t\geq0}.
   \end{align*}
   Notice that in any of the above integrands, the strict inequality of the indicator process can be changed to include an equality, because the Lebesgue measure $ds$ does not charge $\{s;s=\tau\}$.
   From the above equation and since $f(t,X_t)=f(t,X_t)1_{\{\tau\leq t\}}+f(t,X_t)1_{\{\tau> t\}}$, after some manipulations, it concludes that
the process \eqref{eq:compensator-second approach} is an $\mathfrak F^X$- local martingale. Hence  the predictable finite variation part of the process $\left(f(t,X_t)1_{\{\tau> t\}}\right)_{t\geq0}$ is equal to
   \begin{equation*}
      \left(\int_0^t\mathfrak A f(s,X_{s})1_{\{\tau> s\}}\;ds\right)_{t\geq0},
   \end{equation*}
   where $\mathfrak A f(s,x)$ is given by \eqref{eq:compensator-second approach-operator}.
\end{proof}
\begin{remark}
\begin{enumerate}
\item Regarding Theorem \ref{theorem:G-compensator-second approach}, a few comments are worth of mentioning. In the proof of the theorem, the assumption $0<\int_{\mathbb R} x^2\;v(dx)<\infty$ of Hypothesis \ref{hyp:x-integrability1} was not used. The operator $\mathfrak A$ given by \eqref{eq:compensator-second approach-operator} is not the same as Dynkin's or It\^o's operators. Theorem \ref{theorem:G-compensator-second approach} still holds for a $C^{1,1}([0,T]\times\R)$ function $f$ satisfying the integrability condition of Hypothesis \ref{hyp:integrability} on $\mathcal O=[0,T]$, $T>0$. Finally, using Lemmas 3.10 and 3.11 of  Chapter \Rmnum{1} of Jacod and Shiryaev (1987), the canonical decomposition of the theorem shows that the process $(Z_t-Z_0)_{t\geq0}$ belongs to $\mathscr A_{loc}.$
\item    Note that if the derivative of $f$ is bounded, the integrability condition of Definition \ref{assumption:*} is satisfied. In particular, this shows that $\tau$ admits a compensator that is absolutely continuous with respect to the Lebesgue measure; in other words, the following process is an $\mathfrak F^X$- local martingale $$\left(1_{\{\tau\leq t\}}-\int_0^t1_{\{\tau>s\}}v((-\infty,-X_s])\;ds\right)_{t\geq0}.$$ 

Also, for a constant function $f$, and a compound Poisson process plus drift, Theorem \ref{theorem:G-compensator-second approach} is a result of Theorem 1.3 in Guo and Zeng (2008).
\end{enumerate}
\end{remark}

\section{Hedging Strategies for Defaultable Claims}\label{sec:hsdc}

In this section our goal is to obtain locally risk-minimizing hedging
strategies for the credit sensitive security with payoff in \eqref{eq:payoff}.

If the underlying process $X$ is a (local) martingale, local risk-minimization
reduces to risk-minimization and the existence of the hedging strategies is guaranteed by a GKW decomposition. When the process $X$ is a semimartingale then risk-minimization is no longer valid. It must be improved to local risk-minimization and the hedging strategies are solved by
the FS decomposition. 

The FS decomposition was first introduced by
 F\"ollmer and Schweizer (1991). 
 The existence of the FS decomposition of a square-integrable claim is proved even for a $d$-dimensional semimartingale $X$  by Schweizer (1994), assuming that the process $X$ satisfies the SC condition and the mean-variance tradeoff (MVT) process is uniformly bounded in $\omega$ ($\omega$ belongs to $\Omega$), and $t$ and has jumps strictly bounded from above by 1. Monat and
Stricker (1994) prove the existence of the FS decomposition just by assuming that the MVT process
is uniformly bounded in $\omega$ and $t$. Under this condition, further Monat and
Stricker (1995) prove  also the uniqueness.

Choulli, Krawczyk and Stricker (1998) find necessary and sufficient conditions for the existence and uniqueness of the FS decomposition by introducing a new notion for martingales. They prove that there is an FS decomposition for a square-integrable claim under the semimartingale $X=X_0+M+\int\zeta\; d\la M\ra$, if first, the process $\mathcal E (-\int\zeta\; d M)$ satisfies an integrability condition and second if it is ``regular" (we refer to the original paper for a definition). Here the process $\mathcal E (-\int\zeta\; d M)$ is the Dol\'{e}ans - Dade exponential process, see Section 8, Chapter \Rmnum{2} of Protter (2004).

Choulli, Vandaele and Vanmaele (2010) discuss the relationship between the GKW and FS decompositions assuming that $\mathcal E (-\int\zeta\; dM)$ is strictly positive. Then in a general framework, under a weaker assumption that does not require the strict positivity of $\mathcal E (-\int\zeta\; dM)$, they find a closed form of the FS decomposition based on a representation theorem, Theorem 2.1 of their paper. 

Their general framework  can cover our specific model. However in contrast to Theorem 2.1 of their paper, Theorem \ref{theorem:G-compensator-second approach} of our work leads to more explicit solutions for hedging strategies.
In addition, despite current methods that normally start from a payoff and then construct a value process, we somehow turn this around and present self-contained calculations for the components of the FS decomposition. 

Assume that processes $X$ and $Z$ belong to $\mathcal {M}^2_{loc}$ on $[0,T]$. Then by the GKW decomposition there is a predictable process
$\xi$ and a (local) martingale $L$, orthogonal to $X$, such that $$Z=Z_0+\int\xi\;dX+L,$$ and the process $\xi$ is given by
\begin{equation}\label{eq: Strategy}
   \xi=\dfrac{d\la Z,X\ra}{d\la X,X\ra}.
\end{equation}
Also it is worth mentioning that this decomposition is still valid under milder conditions. For instance, it is enough to have $[Z,X]$, $[X]\in\mathscr A_{loc}$, $Z$ a local martingale, and $\xi$ a locally bounded predictable process. In the $\mathcal {M}^2_{loc}$ space, all these conditions are satisfied.


The locally risk-minimizing strategy is linked to the FS decomposition. Hence, our aim is to find the FS decomposition of the payoff \eqref{eq:payoff}. To reach this goal, the next theorem first gives a decomposition that is close to the FS decomposition and in fact is more general. This theorem is also used in Section \ref{sec:edtpt}. Before stating the theorem, we explain the conditions on the underlying process $X$.


Assuming Hypothesis \ref{hyp:x-integrability1} then $\int_{\mathbb R} |x|\;v(dx)<\infty$, and therefore the process $X$ has the canonical decomposition  $X=X_0+M+\Lambda$, where $M$ is a martingale and $\Lambda$
is a continuous finite variation process (in fact a deterministic function) given by $$\Lambda_t=\mu t+\int_0^t\int_{\mathbb{R}} y\;v(dy)\;ds,\quad t\geq0.$$


We remind the reader that the process $\left(f(t,X_t)1_{\{\tau>t\}}\right)_{0\leq t\leq T}$ is represented by $Z=\left(Z_t\right)_{0\leq t\leq T}$. Also let the process $\theta=\left(\theta_t\right)_{0\leq t\leq T}$ be given by
\begin{equation}\label{eq:theta}
   \theta_t=\frac{\mathcal Kf(t,X_{t^-})}{\int_{\mathbb R}y^2\;v(dy)}1_{\{\tau\geq t\}},
\end{equation}
where $\mathcal{K}f(t,x)=\Big(\mathfrak A K(t,x)-x\mathfrak A f(t,x)-\beta f(t,x)\Big)$, the functions $K=K(t,x)$ and $f=f(t,x)$ are defined in Definition \ref{assumption:*}, and $\beta=\mu+\int_{\mathbb R} y\;v(dy)$. Notice that the process $\theta$ is predictable and implicitly depends on the function $F=F(x)$. Also, if $\beta$ is non-zero then $\theta$ can be equivalently represented by $\theta_t=\frac{\mathfrak A f(t,X_{t^-})}{\beta}1_{\{\tau\geq t\}}$.
\begin{theorem}\label{theorem: LHR}
   Assume that Hypothesis \ref{hyp:x-integrability1} holds and let the function $F$ belong to class (*).
   We further suppose that the process $[Z,X]$ belongs to $\mathscr{A}_{loc}$. Then for all $0\leq t\leq T$, the following decomposition holds up to an evanescent set\footnote{This means that up to an evanescent set we have $ Z_.=Z_0+\int\theta\;dX+L.$}
   \begin{equation}\label{eq:function f decomposition}
     Z_t= f(t,X_t)1_{\{\tau>t\}}=Z_0+\int_0^t\theta_{s}\;dX_s+L_t,
   \end{equation}
   and specifically for $t=T$, one obtains
   \begin{equation}\label{eq:lhr-fs-decomposition}
      Z_T=F(X_T)1_{\{\tau>T\}}=Z_0+\int_0^T\theta_{s}\;dX_s+L_T,\quad\text{almost surely,}
   \end{equation}
   where the function $f=f(t,x)$ is introduced in Definition \ref{assumption:*}, and the process $L=\left(L_t\right)_{0\leq t\leq T}$, $L_0=0$, is a local martingale, orthogonal to the martingale part of $X$, i.e. $M$.
\end{theorem}
\begin{proof}
  Since $F=F(x)$ belongs to class (*), by Theorem \ref{theorem:G-compensator-second approach}, there are the following $\left(\mathfrak F^X_t\right)_{0\leq t\leq T}$ - local martingales $M^{(1)}$ and $M^{(2)}$ on $[0,T]$,
   \begin{equation*}
      M_t^{(1)}=Z_t-Z_0-\int_0^t\mathfrak{A}f(s,X_s)1_{\{\tau>s\}}\;ds,
   \end{equation*}
   \begin{equation*}
      M_t^{(2)}=K(t,X_t)1_{\{\tau>t\}}-K(0,X_0)-\int_0^t\mathfrak{A}K(s,X_s)1_{\{\tau>s\}}\;ds.
   \end{equation*}

   First we find the GKW decomposition of $M^{(1)}$ versus $M$. We show that
   \begin{equation}\label{eq:<M^1,M>}
      M^{(1)}_t=\int_0^t\theta_{s}\;dM_s+L_t,\quad 0\leq t\leq T,
   \end{equation}
   for a local martingale $L=\left(L_t\right)_{t\geq0}$ that is orthogonal to $M$.

   By Proposition 4.49, Chapter \Rmnum{1} of Jacod and Shiryaev (1987), $[Z,X]=[M^{(1)},M]$. Therefore $[M^{(1)},M]$ belongs to $\mathscr A_{loc}$ and its compensator exists which is given by $\la M^{(1)},M\ra$, see Section 5, Chapter \Rmnum{3} of Protter (2004). By similar reasons or as we will see shortly, the process $\la M\ra$ also exists. Because of these reasons, the GKW decomposition exists and formula \eqref{eq: Strategy} is applicable. So we need to obtain $\la M^{(1)},M\ra$ and $\la M\ra$.

   Calculating $\la M\ra$ is simple. Since $X$ is square integrable on $[0,T]$, Proposition 4.50, Chapter \Rmnum{1} of Jacod and Shiryaev (1987) shows that $[X]\in\mathscr A_{loc}$. Also, we have that $[M]=[X]$ and therefore the conditional quadratic variation of $M$, as the compensator of $[M]$, exists and equals to $\la X\ra$. Hence the process $\la M\ra$ is equal to
   
   \begin{equation}\label{eq:<M>}
      \la M\ra_t=\int_0^t\int_{\mathbb R}y^2\;v(dy)\;ds.
   \end{equation}

    Since $[M^{(1)},M]=[Z,X]$, the compensator is the same for the two processes and to get $\la M^{(1)},M\ra$, it is enough to obtain $\la Z,X\ra$.
  Integration by parts for semimartingales on $[0,T]$ gives
   \begin{equation*}
      Z_tX_t=Z_0X_0+\int_0^tZ_{s^{-}}\;dX_s+\int_0^tX_{s^{-}}\;dZ_s+[Z,X]_t.
   \end{equation*}
   Let $F^{(1)}_t=\int_0^t\mathfrak{A}f(s,X_s)1_{\{\tau>s\}}\;ds$
    and $F^{(2)}_t=\int_0^t\mathfrak{A}K(s,X_s)1_{\{\tau>s\}}\;ds$, then $Z=Z_0+M^{(1)}+F^{(1)}$, $XZ=X_0Z_0+M^{(2)}+F^{(2)}$, and we also have $X=X_0+M+\Lambda.$
   Therefore the above integration by parts formula on $[0,T]$ becomes
   \begin{eqnarray*}
&& [Z,X]_t-(F^{(2)}_t-\int_0^t X_{s^-}\;dF_s^{(1)}-\int_0^tZ_{s^{-}}\;d\Lambda_s) \\
&&\qquad\qquad\qquad = M_t^{(2)}-\int_0^tX_{s^{-}}dM^{(1)}_s-\int_0^tZ_{s^{-}}dM_s.
   \end{eqnarray*}

  The integrals on the right-hand side of the above equality are local martingales, the process
   $$\left(F^{(2)}_t-\int_0^t X_{s^-}\;dF_s^{(1)}-\int_0^tZ_{s^{-}}\;d\Lambda_s\right)_{0\leq t\leq T}$$ is a
  predictable finite variation process, and $[Z,X]=[M^{(1)},M]$. Therefore the uniqueness of the conditional quadratic covariation
(see Section 5, Chapter \Rmnum{3} of Protter (2004)) gives,
   \begin{equation*}
      \la M^{(1)},M\ra_t=F^{(2)}_t-\int_0^t X_{s^-}\;dF_s^{(1)}-\int_0^tZ_{s^{-}}\;d\Lambda_s,\quad 0\leq t\leq T.
   \end{equation*}
   Hence after some manipulations $\la M^{(1)},M\ra_t$ is seen to be equal to
   \begin{equation}\label{eq: <Z,M>}
      \int_0^t\Big(\mathcal Kf(s,X_{s^-})\Big)1_{\{\tau\geq s\}}\;ds.
   \end{equation}
Then the GKW decomposition in \eqref{eq:<M^1,M>} is a result of expressions \eqref{eq: Strategy}, \eqref{eq:<M>}, and \eqref{eq: <Z,M>}.
By equation \eqref{eq:leb} of Lemma \ref{lem:finitevalue}, $m([0,t]\cap{\{s;X_s=0\}})=0$, almost surely, where $m$ is the Lebesgue measure. This implies that almost surely we have      $M_t^{(1)}=Z_t-Z_0-\int_0^t\mathfrak{A}f(s,X_s)1_{\{\tau>s\}}1_{\{X_s>0\}}\;ds$.
   On the other hand, $f=f(t,x)$ satisfies the PIDE of Definition \ref{assumption:*}, therefore $M_t^{(1)}=Z_t-Z_0-\int_0^t\theta_{s}\;d\Lambda_s$ and the GKW decomposition  \eqref{eq:<M^1,M>} becomes
   \begin{equation}\label{eq:eqref}
      Z_t-\int_0^t\theta_{s}\;d\Lambda_s=f(0,X_0)+\int_0^t\theta_{s}\;dM_s+L_t.
   \end{equation}
   Because functions $f$ and $K$ satisfy the integrability condition of Hypothesis \ref{hyp:integrability}, both integrals $\int_0^t|\mathfrak{A} f(s,X_s)|\;ds$ and $\int_0^t|\mathfrak{A} K(s,X_s)|\;ds$ are almost surely finite for all $0\leq t\leq T$. Therefore for all $0\leq t\leq T$, $\theta_t$ and so the term $\int_0^t\theta_{s}\;d\Lambda_s$ are well defined and almost surely finite. Hence one can move the integral on the left-hand side to the other side of the equality. This gives the decomposition in \eqref{eq:function f decomposition}.
   Finally the decomposition in \eqref{eq:lhr-fs-decomposition} is obtained by letting $t=T$ in equation \eqref{eq:function f decomposition} and noticing that by Definition \ref{assumption:*}, $Z_T=f(T,X_T)1_{\{\tau>T\}}=f(T,X_T)1_{\{\tau>T\}}1_{\{X_T>0\}}=F(X_T)1_{\{\tau>T\}}1_{\{X_T>0\}}$, almost surely.
\end{proof}
\begin{remark}\label{remark:referee}
In the proof of Theorem \ref{theorem: LHR}, the PIDE of Definition \ref{assumption:*} was used at last to obtain equation \eqref{eq:eqref}. On the other hand, the GKW decomposition \eqref{eq:<M^1,M>} is still valid whether or not $f$ and $K$ satisfy this PIDE. In fact, the integrability condition of Hypothesis \ref{hyp:integrability} is all that is needed.
\end{remark}
\begin{remark}
    Note that the calculations of Theorem \ref{theorem: LHR}, especially equation \eqref{eq: <Z,M>} along with Corollary 3.16 of Choulli, Krawczyk and Stricker (1998) show that $Z$ is an $\mathcal E$- local martingale and $Z+\la Z,\widetilde N\ra$ is a local martingale where $\widetilde N=-\dfrac{\beta}{\int_\mathbb{R}y^2\;v(dy)} M$. Then in comparison to Proposition 4.2 of Choulli, Vandaele and Vanmaele (2010), this suggests that $Z$ should be the value of the hedging portfolio. Proposition \ref{prop:hedging strategies} confirms this.
\end{remark}
In the special case when the process $X$ is a local martingale, we have the following corollary.
\begin{corol}\label{corollary: LHR}
    Assume that Hypothesis \ref{hyp:x-integrability1} holds. Let the function $F=F(x)$ belongs to class (*) and the process $[Z,X]$ belong to $\mathscr{A}_{loc}$. Now further suppose that $X$ is 
   a local martingale in the natural completed filtration generated by $X$, i.e. $\mathfrak F^X$. Then we have
   \begin{equation}
      Z_T=F (X_T)1_{\{\tau>T\}}=Z_0+\int_0^T\frac{\mathfrak A K(s,X_{s^{-}})}{\int_{\mathbb R}y^2\;v(dy)}1_{\{\tau\geq s\}}\;dX_s+L_T,\quad\text{ almost surely,}
   \end{equation}
   where the operator $\mathfrak A$ is introduced in \eqref{eq:compensator-second approach-operator}, the functions $f=f(t,x)$ and $K=K(t,x)$ are defined in Definition \ref{assumption:*},  and the process $L=\left(L_t\right)_{0\leq t\leq T}$, $L_0=0$ is a local martingale orthogonal to $X$.
\end{corol}
\begin{proof}
   Since $X$ is a local martingale, then $\beta=\mu+\int_{\mathbb R} y\;v(dy)$ is equal to zero, and therefore by Definition \ref{assumption:*}, $\mathfrak A f$ is also zero. Now the corollary easily follows from Theorem \ref{theorem: LHR}.
\end{proof}

Our goal is to find the FS decomposition of the payoff in  \eqref{eq:payoff}, but finding this decomposition leads to just a pseudo-local risk-minimization and not necessarily to local risk-minimization. To make a bridge between the
two concepts, first we need to investigate the SC condition on the
underlying process and also the existence of the FS decomposition, see Schweizer (2001) for more details.
Since the process $X$ satisfies Hypothesis \ref{hyp:x-integrability1}, it is square integrable and one can easily prove that the SC condition holds for $X$. 

Therefore, by Theorem 3.3 of Schweizer (2001), locally risk-minimizing
strategies are the same as pseudo-locally risk-minimizing strategies. On
the other hand by Proposition 3.4 of Schweizer (2001) the existence
of the latter is equivalent and a result of the existence of the FS
decomposition of the payoff. Since the MVT process
in our model is uniformly bounded in both $t$ and $\omega$, the
FS decomposition exists.

So we conclude that in our framework the existence of the
F\"ollmer-Schweizer decomposition, and so locally risk-minimizing
strategies, are guaranteed.

From the above, to get a local risk-minimization
strategy, all we need is to find the FS decomposition. Some integrability conditions turn the decomposition in \eqref{eq:lhr-fs-decomposition} into the FS decomposition. The next proposition clarifies this. First we provide the definition of $\Theta$, $L^2$- strategy, and $L^2(X)$ that are used in the following proposition, see Schweizer (2001) for more explanations.
\begin{Def}
     Assume that $X$ is a local martingale. Then $L^2(X)$ is the space of all predictable processes $\theta$ such that $\E[\int_0^T\theta_s^2\;d[X]_s]<\infty$.
\end{Def}
\begin{Def}
     Assume that $X$ is a square integrable special semimartingale with the canonical decomposition $X=X_0+M+A$.  Then $\Theta$ is the
     space of all predictable processes $\theta$ such that $\E\left[\int_0^T\theta_s^2\;d[M]_s+(\int_0^T|\theta_s\;dA_s|)^2\right]<\infty$.
\end{Def}
\begin{Def}
   An $L^2$-strategy is a pair $\phi= (\theta,\eta)$, where $\theta\in\Theta$ and $\eta$ is a real valued adapted process such that the value process $V(\phi)= \theta X +\eta$ is right-continuous and square-integrable. That means $V_t(\phi)\in L^2(\Omega,\mathfrak F_t,\pr)$ for each $t\in[0,T]$.
\end{Def}

\begin{prop}\label{prop:hedging strategies}
    Assume that Hypothesis \ref{hyp:x-integrability1} holds and let the function $F=F(x)$ belongs to class (*).
   We further suppose that for all $0\leq t\leq T$, $f(t,X_t)$ belongs to $L^2(\Omega,\mathfrak F_t,\pr)$ and the process $\theta$ given by \eqref{eq:theta} is in $\Theta$. Then there exists a locally risk-minimizing $L^2$- strategy $\phi=(\theta,\eta)$ as follows. The number of shares to be invested in the risky asset is given by $\theta$. The hedging error $L$ belongs to $\mathcal M^2_{0}$. It is orthogonal to $M$ and given by
   $$L_t=Z_t-Z_0-\int_0^t\theta_s\;dX_s,\quad 0\leq t\leq T.$$ The value process of the portfolio $\left(V_t(\theta)\right)_{t\geq0}$ associated with the strategy $\phi$ is equal to
   $$V_t(\theta)=Z_0+\int_0^t\theta_s\;dX_s+L_t,\quad 0\leq t\leq T,$$
   the number of shares to be invested in the risk-free asset is
   $$\eta_t=V_t(\theta)-\theta_t X_t,\quad 0\leq t\leq T,$$
   and finally the cost process is given by
   $$C_t=Z_0+L_t,\quad 0\leq t\leq T.$$
\end{prop}
\begin{proof}
    The process $X$ satisfies the SC condition. Therefore the existence of an $L^2$- strategy is equivalent to the existence of the FS decomposition. Notice that for all $0\leq t\leq T$, $f(t,X_t)$ belongs to $L^2(\Omega,\mathfrak F_t,\pr)$, and so by Proposition 4.50, Chapter \Rmnum{1} of Jacod and Shiryaev (1987), the process $[Z,X]$ is in $\mathscr{A}_{loc}$. From equation \eqref{eq:function f decomposition} of Theorem \ref{theorem: LHR}, we have
    \begin{equation*}
       Z_t-\int_0^t\theta_{s}\;d\Lambda_s=Z_0+\int_0^t\theta_{s}\;dM_s+L_t,\quad 0\leq t\leq T,
    \end{equation*}
    where $L$ is a local martingale orthogonal to $M$.
    Because $\theta$ is in $\Theta$ and $f(t,X_t)$ is square-integrable, the left-hand side and so the right-hand side of the above equation are square-integrable.
    Since $\theta$ belongs to $\Theta$, it is also in $L^2(M)$ and so by Lemma 2.1 of Schweizer (2001) the process $\int\theta\; dM$ is in $\mathcal M^2_0.$ Hence the process $L$ is square-integrable on $[0,T]$ and belongs to $\mathcal M^2_0$. Now the result follows from Proposition 3.4 of Schweizer (2001).
\end{proof}

\begin{remark}
   A similar result to Proposition \ref{prop:hedging strategies} can be obtained when $X$ is a local martingale, but with a simpler form for the strategy $\theta$. Notice that although we did not use the MELMM method, we have paid the price by involving a PIDE. In the MELMM method when the underlying process is a martingale the problem of finding the hedging strategies is simpler. Here, the same happens, if the underlying process is a martingale, the PIDE to solve for the hedging strategy has a simpler form.
\end{remark}

The next theorem investigates necessary and sufficient conditions under which the process $L$ in Theorem \ref{theorem: LHR} vanishes. 

\begin{theorem}\label{theorem:L=0}
   Assume that Hypothesis \ref{hyp:x-integrability1} holds and the function $F=F(x)$ belongs to class (*). Suppose that the integrability condition of Hypothesis \ref{hyp:integrability} is met on $\mathcal O=[0,T]$ by the function $f^2$, defined as $f^2(t,x)=(f(t,x))^2$, where the function $f$ is defined in Definition \ref{assumption:*}. Now further suppose that the process $[Z,X]$ and the process $[L]$ in the decompositions \eqref{eq:function f decomposition} and \eqref{eq:lhr-fs-decomposition} belong to $\mathscr A_{loc}$. 
   Let the operator $\mathfrak L$ be defined as
   \begin{equation}\label{eq:operator L}
      \mathfrak L f(t,x)=\mathfrak{A}f^2(t,x)-2\beta f(t,x)-\frac{\Big(\mathcal K f(t,x)\Big)^2}{\int_{\mathbb R}y^2\;v(dy)},
   \end{equation}
   where $\mathcal{K}f(t,x)=\Big(\mathfrak A K(t,x)-x\mathfrak A f(t,x)-\beta f(t,x)\Big)$
    and the function $K=K(t,x)$ is defined in Definition \ref{assumption:*}. Then the martingale $L$ is null on $[0,T]$, if and only if $\mathfrak L f(t,x)=0$ for all $0\leq t\leq T$ and all $x>0$. In this case, for all $0\leq t\leq T$, we have the following, up to an evanescent set,
   \begin{equation}\label{eq:function f decomposition-risk free}
     Z_t= f(t,X_t)1_{\{\tau>t\}}=Z_0+\int_0^t\theta_{s}\;dX_s,
   \end{equation}
   and specifically for $t=T$, one obtains
   \begin{equation}\label{eq:lhr-fs-decomposition-risk free}
      Z_T=F(X_T)1_{\{\tau>T\}}=Z_0+\int_0^T\theta_{s}\;dX_s,\quad\text{almost surely}.
   \end{equation}
\end{theorem}
\begin{proof}
   Since $[L]$ is
   in $\mathscr A_{loc}$, by Corollary \ref{corol: <x,x>=0}, $L=0$ is equivalent to $\la L,L\ra=0$. On the other hand, by Theorem \ref{theorem: LHR} the following holds
   \begin{equation*}
      Z_t=Z_0+\int_0^t\theta_{s}\;dX_s+L_t.
   \end{equation*}
   From this decomposition, we have
   \begin{equation}\label{eq: <L,L>}
      \la L,L\ra=\la Z\ra-2\la
      Z,\int\theta_{}\;dX\ra+\la\int\theta_{}\;dX\ra.
   \end{equation}
In what follows, we show that this equation is valid, in the sense that all the terms on the right
hand side exist and we compute them explicitly. First, let us obtain $\la Z\ra.$

We already know that $Z=Z_0+M^{(1)}+F^{(1)}$ and observe that
$Z^2_t=f^2(t,X_t)1_{\{\tau>t\}}$. By Theorem \ref{theorem:G-compensator-second approach}, $Z^2=Z_0^2+M^{(3)}+F^{(3)}$, where $M^{(3)}$ is
an $\mathfrak F^X$- local martingale and $F^{(3)}_t=\int_0^t\mathfrak{A}f^2(s,X_s)1_{\{\tau>s\}}ds.$ Using the integration by parts formula we get that
\begin{equation*}
   Z^2=Z_0^2+2\int Z_{^-}\;dM+2\int Z_{^-}\;d\Lambda+[Z],
\end{equation*}
\begin{equation*}
   Z_0^2+M^{(3)}+F^{(3)}=Z_0^2+2\int Z_{^-}\;dM+2\int Z_{^-}\;d\Lambda+[Z],
\end{equation*}
or
\begin{equation*}
   [Z]-(F^{(3)}-2\int Z_{^-}\;d\Lambda)=M^{(3)}-2\int Z_{^-}\;dM.
\end{equation*}
The right-hand side of the above equation is a local martingale. This shows that $[Z]\in\mathscr A_{loc}$. Now the
predictability of $(F^{(3)}-2\int Z_{^-}\; d\Lambda)$ and uniqueness of conditional quadratic variation give
\begin{equation*}
   \la Z\ra_t=\int_0^t\mathfrak{A}f^2(s,X_s)1_{\{\tau>s\}}\;ds-2\int_0^t Z_s\;d\Lambda_s.
\end{equation*}

For the second term of \eqref{eq: <L,L>}, since $[Z,X]=[M^{(1)},M]$ and $[Z,X]\in \mathscr A_{loc}$, computing the second term follows from
$$\la Z,\int\theta\;dX\ra_t = \int_0^t\theta_s\;d\la M^{(1)},M\ra_s
	= \int_0^t\frac{\Big(\mathcal K f(s,X_s)\Big)^2}
	{\int_{\mathbb R}y^2\;v(dy)}1_{\{\tau> s\}}\;ds,$$
where $\la M^{(1)},M\ra$ was already computed in the proof of Theorem \ref{theorem: LHR}, see equations \eqref{eq:theta} and \eqref{eq: <Z,M>}.

The process $[X]$ belongs to $\mathscr A_{loc}$ and the third term can be computed similarly
$$\la\int\theta\;dX\ra_t=\int_0^t\theta^2\;d\la M\ra,$$
or
\begin{align*}
    \la\int\theta\;dX\ra_t= \int_0^t\frac{\Big(\mathcal K f(s,X_s)\Big)^2}{\int_{\mathbb R}y^2\;v(dy)}1_{\{\tau> s\}}\;ds.
\end{align*}

From \eqref{eq: <L,L>} and the previous calculation we get
the following
$$\la L,L\ra_t=\int_0^t\mathfrak{L}f(s,X_{s})1_{\{\tau>s\}}1_{\{X_s>0\}}\;ds,\quad\text{almost surely},$$
where
\begin{align*}
    \mathfrak{L}f(t,x)= \mathfrak{A}f^2(t,x)-2\beta f(t,x)-\frac{\left(\mathcal K f(s,x)\right)^2}{\int_{\mathbb R}y^2\;v(dy)}.
\end{align*}

Since the function $f$ is  $C^{1,1}$, $\la L,L\ra$ is zero almost surely on $[0,T]$ if and only if $\mathfrak{L} f(t,x)=0$ on $[0,T]\times\mathbb R^+$. On the other hand by Corollary \ref{corol: <x,x>=0}, the former is equivalent to $L=0.$
Therefore in the decompositions \eqref{eq:function f decomposition} and \eqref{eq:lhr-fs-decomposition}, the orthogonal part vanishes if and only if $\mathfrak{L} f(t,x)=0$ on $[0,T]\times\mathbb R^+$ and this gives equations \eqref{eq:function f decomposition-risk free} and \eqref{eq:lhr-fs-decomposition-risk free}.
\end{proof}

By combining Theorem \ref{theorem:L=0} and Proposition \ref{prop:hedging strategies}, we get the following result that provides a necessary and sufficient condition for the existence of a risk-free product. In the context of jump-diffusion processes, Kunita (2010) answers a similar question for path independent payoffs.
\begin{prop}\label{prop:risk-free product}
   Assume that Hypothesis \ref{hyp:x-integrability1} holds and the function $F=F(x)$ satisfies Definition \ref{assumption:*}. Suppose that the integrability condition of Hypothesis \ref{hyp:integrability} is met on $\mathcal O=[0,T]$ by function $f^2$ defined as $f^2(t,x)=(f(t,x))^2$, where the function $f$ is defined in the hypothesis. Now further suppose that for all $0\leq t\leq T$, $f(t,X_t)$ belongs to $L^2(\Omega,\mathfrak F_t,\pr)$ and the process $\theta$ given by \eqref{eq:theta} is in $\Theta$.
   Let the operator $\mathfrak L$ be defined as \eqref{eq:operator L}.
   Then the process $\phi=(\theta,\eta)$, defined in Proposition \ref{prop:hedging strategies}, is a locally risk-minimizing $L^2$- strategy  that makes the derivative $F(X_T)1_{\{\tau>T\}}$ completely hedgeable if and only if $\mathfrak L f(t,x)=0$, for all $0\leq t\leq T$, and all $x>0$. It means that we have the following decomposition
   \begin{equation*}
      F(X_T)1_{\{\tau>T\}}=f(0,X_0)+\int_0^T\theta_{s}\;dX_s.
   \end{equation*}
\end{prop}
\begin{remark}
   If the process $X$ is a martingale, then the operator $\mathfrak L$ simplifies to
   \begin{equation*}
      \mathfrak L f(t,x)=\mathfrak{A}f^2(t,x)-\frac{\left(\mathfrak A K(t,x)\right)^2}{\int_{\mathbb R}y^2\;v(dy)}.
   \end{equation*}
\end{remark}
\begin{remark}
     Based on Theorem \ref{theorem:L=0} and Proposition \ref{prop:risk-free product}, the nullity of the martingale $L$ depends on the existence of
     two functions $F=F(x)$ and $f=f(t,x)$, such that they simultaneously satisfy both Definition  \ref{assumption:*} and $\mathfrak L f(t,x)=0$, $x>0$, i.e. a system of PIDE equations. The existence of such a solution is an open problem for the authors.
\end{remark}
The next example shows an application of Proposition \ref{prop:hedging strategies}, and it was chosen simply because in this case there is a closed form solution for the PIDE of Definition \ref{assumption:*} given by Theorem
5.6.3 of Rolski et al.~(1999). Therefore, one can compare the simulated solution with the exact one. In general,
simulation techniques must be used.
\begin{exap}\label{ex:example1}
   Assume that $X_t=u+\mu t+\sum_{j=1}^{N_t}Y_i$, where $N$ is an homogeneous Poisson process with intensity $\lambda$ and the $Y_i$'s are i.i.d.~random variables with jump distribution $F_Y$. Let $\mu>0$, $-Y_1\sim exponential(\delta)$, and suppose that the process $X$ is a martingale in the natural filtration generated by $X$, which means that $\lambda=\mu\delta$. We remind the reader that in this paper the interest rate is taken to be zero. Consider a defaultable zero-coupon bond that pays one unit of currency if there is no default, i.e. $F(x)=1$ for all $x$. One can check that $F$ belongs to class (*) and by Proposition \ref{prop:hedging strategies}, to implement the hedging strategy, the number of shares invested in the risky asset is given by
   \begin{equation}\label{eq:f(t,x)-exp}
      \theta_s=\left(\delta^2\int_{-X_{s^{-}}}^{0}yf(s,X_{s^{-}}+y)F_Y(dy)+\delta f(s,X_{s^{-}})\right)1_{\{\tau\geq s\}},
   \end{equation}
   where $f=f(t,x)$ satisfies the following PIDE
   \begin{equation*}
     \mathfrak A f(t,x)=0,\;\text{for all}\;0\leq t\leq T\;\text{and all}\;x>0,
   \end{equation*}
   $$f(T,x)=1,\;\text{for all}\;x>0.$$
   The Feynman-Kac formula or a renewal argument can be applied to prove that the solution has the following representation $$f(t,x)=1-\pr(\tau\leq T-t|X_0=x).$$ This representation holds regardless of the type of distribution of jumps. In the case of exponential jumps in this example, a closed form solution is available. This solution is provided by Theorem 5.6.3 or Theorem  5.6.4 in Rolski et al.~(1999). It is a complicated function and its graph on $[0,2]\times[0,0.4]$ is given in Figure
   \ref{fig:The exact graph of the function f for exponential jump size distribution} for $\mu=0.1$, $\delta=100$, $\lambda =10$, and $T=2$.

      \begin{figure}[ht]
         \centering
           \includegraphics[scale=0.5]{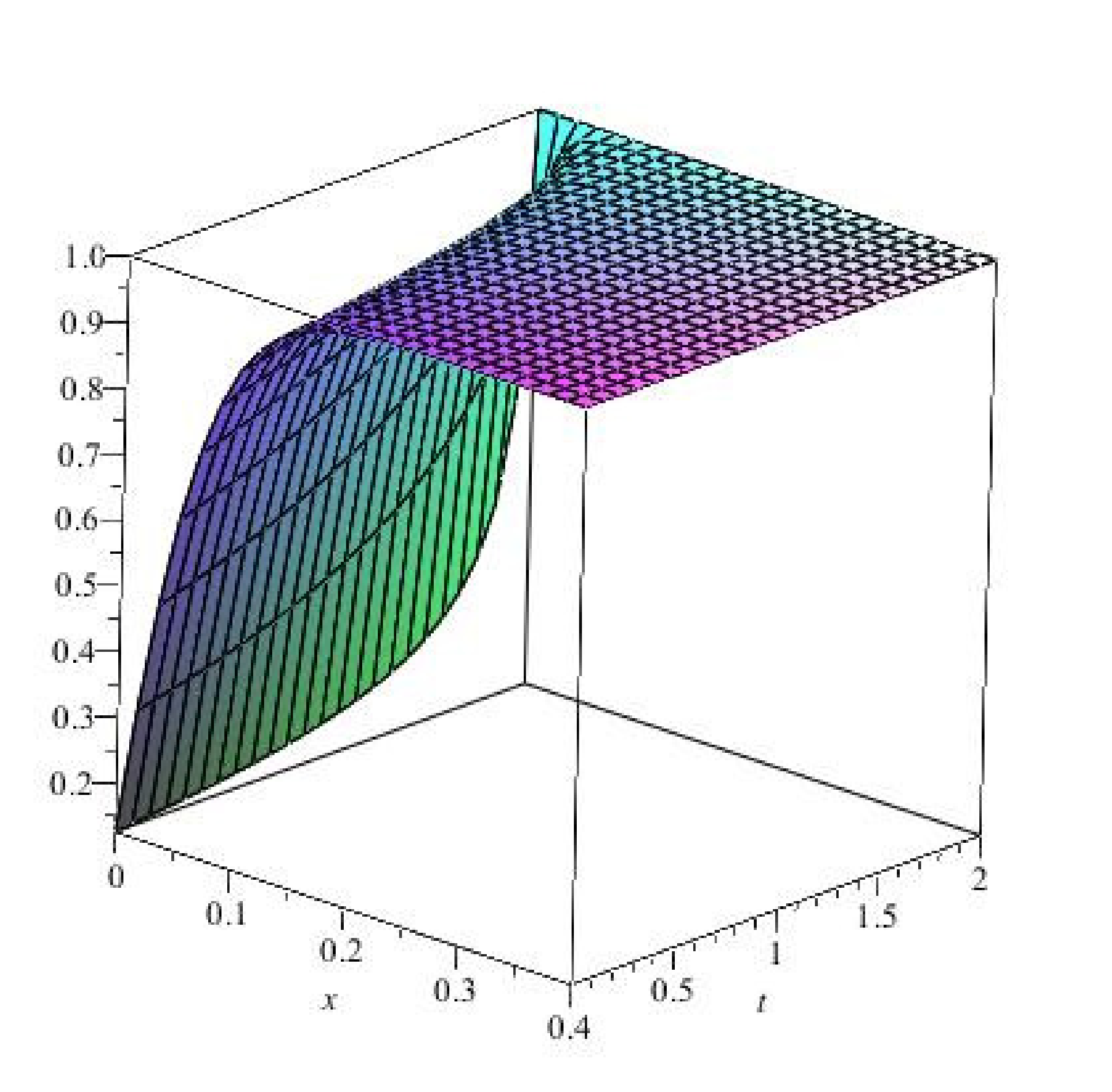}
         \caption{The exact function $f$ with exponential jumps.}
         \label{fig:The exact graph of the function f for exponential jump size distribution}
      \end{figure}

   The function $f=f(t,x)$ can also be estimated numerically by simulation. With the same parameters as above,
   Figure \ref{fig:The estimated graph of the function f for exponential jump size distribution} gives the graph of an estimation of $f=f(t,x)$ on $[0,2]\times[0,0.4]$. The number $\theta$ of shares invested in the risky asset is a closed form of this function given by equation \eqref{eq:f(t,x)-exp}. Therefore the function $f=f(t,x)$ acts as an interface to solve the problem. However this function also has a nice interpretation. From Proposition \ref{prop:hedging strategies}, one can easily verify that the value process of the portfolio is provided through the function $f=f(t,x)$. More precisely we have that $V_t(\theta)=f(t,X_t)1_{\{\tau>t\}}.$
      \begin{figure}[ht]
         \centering
           \includegraphics[scale=0.6]{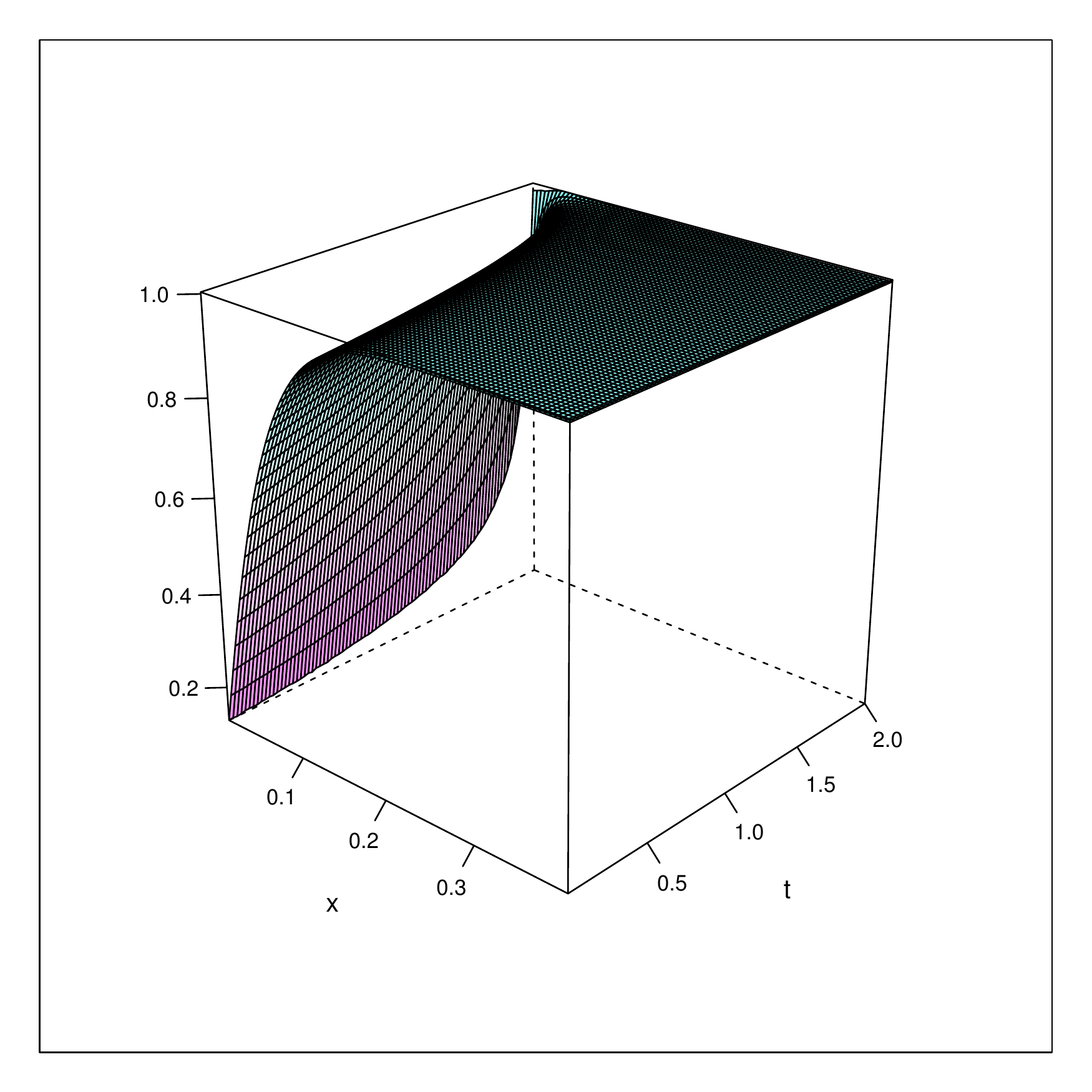}
         \caption{The estimated function $f$ with exponential jumps.}
         \label{fig:The estimated graph of the function f for exponential jump size distribution}
      \end{figure}

      Next we obtain the locally risk-minimizing strategy corresponding to a simulated sample path of the process $X$. In practice, a dynamic portfolio is updated at some specific trading dates. In fact Proposition \ref{prop:hedging strategies} and formula \eqref{eq:f(t,x)-exp} cannot be applied directly. A discretization procedure is required to implement the theory.

      Here we use a simple procedure. We divide the interval $[0,T]=[0,2]$ into 1000 equal subintervals. It is assumed that the trading dates are given by $\{t_0,t_1,...,t_{1000}\}$, for $t_j=\frac{jT}{1000}$, where $j=0,1,...,1000$. Then the number of shares invested in the risky asset is given by
      $$\theta_t=\theta_01_{t=0}+\sum_{k=0}^n\theta_k1_{(t_k,t_{k+1}]}(t),$$
      where each $\theta_k$ is a bounded $\mathfrak F^X_{t_i}$-measurable random variable that is determined right after the transaction $t_k$. This is due to the fact that a realistic strategy must be left continuous or predictable. The integral $\int\theta\;d X$ also must be discretized. This is essential to obtain the observed values of the process $L$.

      Figure \ref{fig:The strategy corresponding to a sample path of the process X} shows the simulated sample path of the process $X$, for $X_0=0.01$, together with the number $\theta$ of shares invested in the risky asset to be held in each trading period. As Figures \ref{fig:The exact graph of the function f for exponential jump size distribution} or \ref{fig:The estimated graph of the function f for exponential jump size distribution} confirm, the probability of crossing the barrier  is relatively high for this process, $\pr(\tau\leq 2)\approx0.754995$. For the sample path of the process $X$ shown in Figure \ref{fig:The strategy corresponding to a sample path of the process X}, the default happens at $\tau\approx0.30869$. At this time, the number $\theta$ drops to zero and remains in this state until the maturity of the contract.
      \begin{figure}
         \centering
         \begin{tabular}[ht]{c c}
            \includegraphics[scale=0.49]{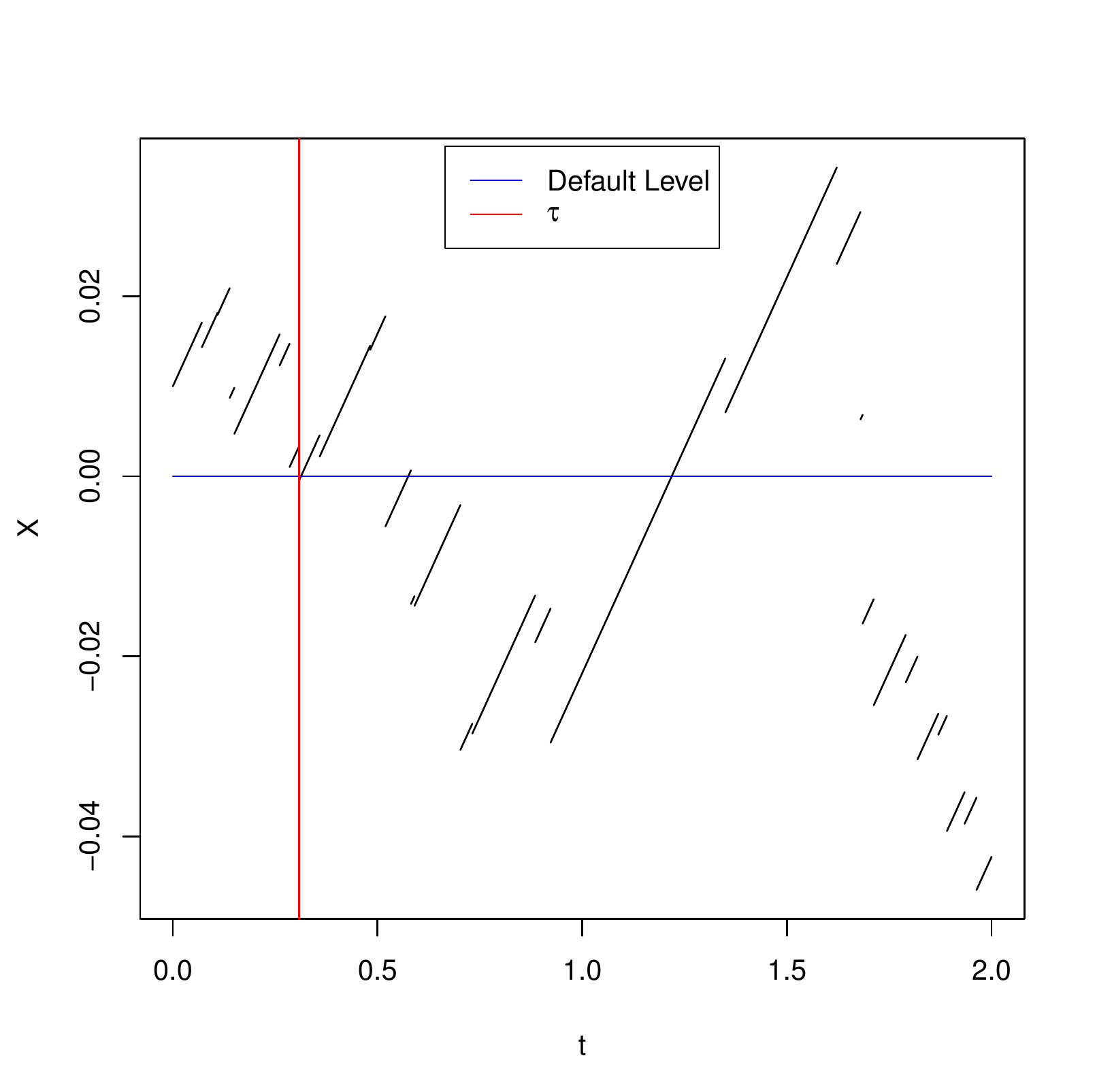}&\includegraphics[scale=0.45]{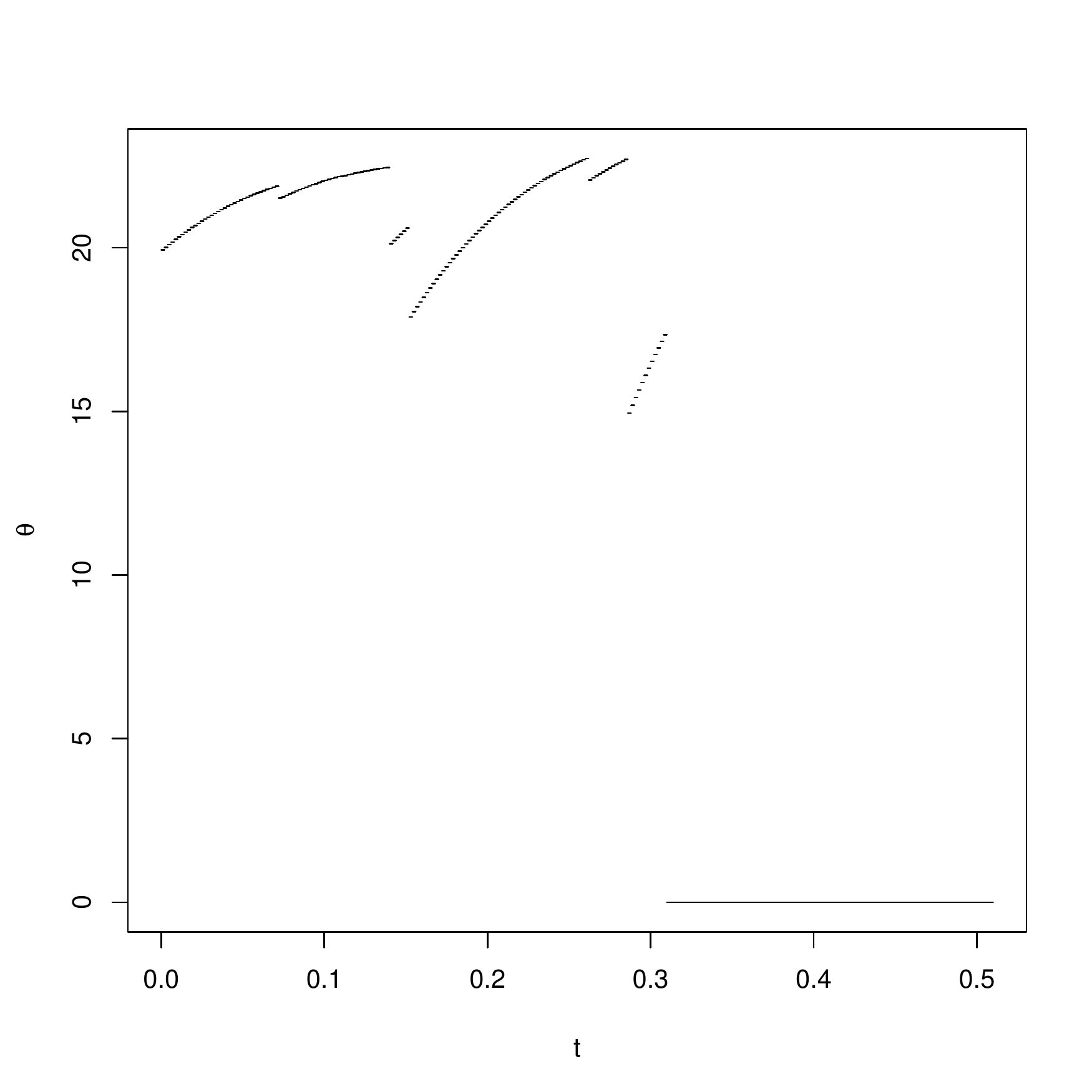}\\
            A sample path of $X$ & The number $\theta$ of shares\\
         \end{tabular}
         \caption{Sample paths of the processes $\theta$ and $X$.}
         \label{fig:The strategy corresponding to a sample path of the process X}
      \end{figure}
       Similar graphs can be obtained for the number $\eta$ of shares of the risk-free asset, the value of the portfolio $V(\theta)$, the error term $L$, and the cost process $C$.
\end{exap}

\section{Structure and Distribution of the Default Time}
\setcounter{equation}{0}\label{sec:edtpt}

In this section, we
discuss the structure and the distribution of the default time. It is assumed that $X$ satisfies Hypothesis \ref{hyp:x-integrability1}. 


Some of the results of Section \ref{sec:hsdc}
can be helpful to understand the structure of the default time. Regarding Theorem \ref{theorem: LHR}, one can let $F=F(x)$ be the constant function $F=1.$ So without almost any effort, we have the following decomposition.
\begin{prop}\label{prop:default time-decomposition1}
  Assume that Hypothesis \ref{hyp:x-integrability1} holds.
   Then for all $0\leq t\leq T$, we have the following decomposition up to an evanescent set
   \begin{equation}\label{eq:default time decomposition1-1}
      Z_t=f(t,X_t)1_{\{\tau>t\}}=Z_0+\int_0^t\theta_{s}\;dX_s+L_t,
   \end{equation}
   and specifically for $t=T$, one obtains
   \begin{equation}\label{eq:default time decomposition1-2}
      1_{\{\tau>T\}}=f(0,X_0)+\int_0^T\theta_{s}\;dX_s+L_T,\quad\text{almost surely,}
   \end{equation}
   where the function $f=f(t,x)$ is introduced in Definition \ref{assumption:*}, the process $\theta=\left(\theta_t\right)_{0\leq t\leq T}$ is given by equation \eqref{eq:theta},
   and the process $L=\left(L_t\right)_{0\leq t\leq T}$, $L_0=0$, is a local martingale orthogonal to the martingale part of $X$, i.e. $M$.
\end{prop}
Notice that since the process $X$ is square-integrable, the process $\left([1_{\{\tau>t\}},X_t]\right)_{0\leq t\leq T}$ belongs to $\mathscr{A}_{loc}$. Although this decomposition reveals the structure of the default time, it does not tell us much about the distribution of the default time. This is the decomposition of the indicator process versus the process $X$. Regarding the distribution of the default time, a more useful decomposition is the following.
\begin{prop}\label{prop:default time-decomposition2}
   Assume that Hypothesis \ref{hyp:x-integrability1} holds. Let the $C^{1,1}$ function $f=f(t,x)$ be the solution of the following PIDE, $$\mathfrak A f(t,x)=0,\;\text{for all}\;0\leq t\leq T\;\text{and all}\;x>0,$$ $$f(T,x)=F(x),\;\text{for all}\;x>0,$$ where
    the function $F=F(x)$ is a real valued function and the function $f$ satisfies the integrability condition of Hypothesis \ref{hyp:integrability}. Let the process $M$ be the martingale part of the canonical decomposition of $X$, i.e. $X=X_0+M+A$.  We further suppose that the process $[Z,X]$ belongs to $\mathscr{A}_{loc}$. Then for all $0\leq t\leq T$, the following decomposition holds up to an evanescent set
   \begin{equation*}
     Z_t= f(t,X_t)1_{\{\tau>t\}}=Z_0+\int_0^t\theta_{s}\;dM_s+L_t,
   \end{equation*}
    and especially for $t=T$, one obtains
   \begin{equation}\label{eq:default time-decomposition2-1}
      Z_T=F(X_T)1_{\{\tau>T\}}=Z_0+\int_0^T\theta_{s}\;dM_s+L_T,\quad\text{almost surely,}
   \end{equation}
   where the process $\theta$ is given by
   $$\theta_t=\frac{\Big(\mathfrak A K(t,X_{t^{-}})-\beta f(t,X_{t^{-}})\Big)}{\int_{\mathbb R}y^2\;v(dy)}1_{\{\tau\geq t\}},$$
   and the process $L=\left(L_t\right)_{0\leq t\leq T}$, $L_0=0$, is a local martingale orthogonal to the process $M$.
\end{prop}
\begin{proof}
The result basically follows from, Remark \ref{remark:referee}, equation \eqref{eq:<M^1,M>}, and then by simplifying equation \eqref{eq:theta} using $\mathfrak A f(t,x)=0$.
\end{proof}
As a special case let $F=1$, then by taking the expectation of both sides of \eqref{eq:default time-decomposition2-1}, we obtain $\pr(\tau>T)=f(0,X_0)$, and $f=f(t,x)$ is the solution of the PIDE in Proposition \ref{prop:default time-decomposition2}. Finding the distribution of the default time using a PIDE is already known; for example, see Theorem 11.3.3 and its proof in Rolski et al.~(1999) where this PIDE is obtained for a compound Poisson process plus drift. 
\begin{exap}\label{example:finite horizon ruin time}
Let $X$ be the same process as Example \ref{ex:example1} and define the function $F=F(x)$ by $F(x)=1-\frac{\lambda}{\mu\delta} e^{(\frac{\lambda}{\mu}-\delta)x}$. Apply Proposition \ref{prop:default time-decomposition2} for $f(t,x)=F(x)$, then $\mathfrak A f(t,x)=0$, and
   $$F(X_t)1_{\{\tau>t\}}=F(u)+\int_0^t\theta_{s}\;dM_s+L_t,$$
   $$\theta_s=\left(\delta^2\int_{-X_{s^{-}}}^{0}yf(s,X_{s^{-}}+y)\;F_Y(dy)+\delta f(s,X_{s^{-}})\right)1_{\{\tau\geq s\}}.$$
   Note that the above function $F=F(x)$ is a special one that makes the operator $\mathfrak A$ zero and hence the process $(F(X_t)1_{\{\tau>t\}})_{t\geq0}$
   is a martingale. This martingale can also be obtained from Theorem \ref{theorem:G-compensator-second approach}. Therefore we have the following identity $$\pr(\tau>t)-\frac{\lambda}{\mu\delta}\E[e^{(\frac{\lambda}{\mu}-\delta)X_t}1_{\{\tau>t\}}]=F(u).$$
\end{exap}

\section*{Conclusion}
\setcounter{equation}{0}\label{sec:c}
In this paper, first a canonical decomposition of the processes $\left(f(t,X_t)1_{\{\tau>t\}}\right)_{t\geq0}$ was studied under some conditions. Then based on this result, the locally risk-minimizing approach was carried out to obtain hedging strategies for certain structural defaultable claims under finite variation L\'evy processes.

The analysis is done simultaneously, when the underlying process has jumps, the security is linked to a default event, and the probability measure is a physical one. This approach does not use the MELMM method or any type of Girsanov's theorem to obtain the strategies. However, the final answer is based on the solution of a PIDE. Besides, some theoretical results in finite horizon ruin time were obtained.

\section*{Acknowledgments}

The authors are grateful to the Associate Editor and anonymous referees for their constructive comments. The first author is also very thankful to Friedrich Hubalek and Julia Eisenberg for many useful discussions.

\appendix
\section{Technical Results}\label{appendix: technical results}
In what follows the concept of creeping and some technical results are discussed.
\begin{Def}\label{Def:creeping}
   Assume that the process $X$ is a L\'{e}vy process such that $X_0=0$ (resp. $X_0=u>0$). Let the stopping time $\tau^+$ be defined as $$\tau^+=\inf\{t>0;X_t>x\}.$$ Then $X$ creeps over (resp. creeps down) the level $x>0$ (resp. $x=0$), when
   $$\pr(X_{\tau^+}=x) > 0\qquad(\text{resp. $\pr(X_{\tau}=0) > 0$, where $\tau$ is defined by \eqref{eq:default}}).$$
\end{Def}
The following theorem is part (i) of Theorem 7.11 of Kyprianou (2006) that gives necessary and sufficient conditions for a process to creep upwards or creep downwards.
\begin{theorem}\label{theorem:creeping}
   Suppose that $X$ is a bounded variation L\'{e}vy process which is not a compound
   Poisson process with the characteristic exponent $\Psi(\theta):=-\log\expect{e^{i\theta X_1}}$. Then $X$ creeps upwards (resp. downwards) if and only if the process $X$ has the following L\'{e}vy-Khintchine exponent
   $$\Psi(\theta) = -i\theta\mu +\int_{\mathbb R-\{0\}}(1-e^{i\theta x})\;v(dx),$$
   for $\mu > 0$ (resp. $\mu<0$), and $v$ is the L\'{e}vy measure.
\end{theorem}
\begin{remark}\label{remark:tist}
Since the process $X$ in Hypothesis \ref{hyp:x-integrability1} is not a compound Poisson process and its drift is positive, Theorem \ref{theorem:creeping} is in force and the process $X$ never creeps down. In simple words, this guarantees that the default happens only by a sudden jump of the process $X$.
 Hence from Meyer's previsibility theorem (see Theorem 4, Chapter \Rmnum{3} of Protter (2004)), the default time $\tau$, given by \eqref{eq:default}, is a totally inaccessible stopping time. 
\end{remark}

\begin{lemma}\label{lem:compensator-second approach}
Let $A$ and $A^p$ belong to $\mathscr A_{loc}$, the class of processes with locally integrable variation. Assume that $\expect{\integral{0}{\infty}{H}dA}=\expect{\integral{0}{\infty}{H}dA^p}$ for all predictable processes $H$ that are non-negative and bounded (in the sense that for each such predictable process $H$, there is an upper bound $c$ free from $t$ and $\omega$ such that $H_t(\omega)\leq c$ for all $t\geq0$ and $\omega\in\Omega$). Then $A-A^p$ is a local martingale.
\end{lemma}
\begin{lemma}\label{lem:finitevalue}
     Assume that the process $X$ satisfies Hypothesis \ref{hyp:x-integrability1}, and $\tau$ is given by \eqref{eq:default}. Then for every $t\geq0$, $\int_0^t\int_{(-\infty,-X_s]}1_{\{\tau>s\}}\;v(dy)\;ds$ is well-defined and almost surely finite.
\end{lemma}
\begin{proof}
Let $m$ be the Lebesgue measure. Since $m$ and $v$ are $\sigma$-finite measure, by Fubini-Tonelli Theorem, the integral $\int1_{(s,y)\in[0,t]\times(-\infty,-X_s]}1_{\{\tau>s\}}\;m\times v(ds\times dy)$ (henceforth denoted by $A$) is well-defined and equal to $\int_0^t\int_{(-\infty,-X_s]}1_{\{\tau>s\}}\;v(dy)\;ds$. Now, we prove its finiteness in three steps.

\textbf{Step 1.} Since $1_{\{\tau>s\}}=1_{\{\tau>s\}}1_{\{X_s\geq0\}}$, we get
\begin{align*}
 A&= \int1_{(s,y)\in[0,t]\times(-\infty,-X_s]}1_{\{\tau>s\}}1_{\{X_s\geq0\}}\;m\times v(ds\times dy)
\\  &= \int1_{(s,y)\in[0,t]\times(-\infty,-X_s]}1_{\{\tau>s\}}1_{\{X_s=0\}}\;m\times v(ds\times dy)\\
          &\quad+\int1_{(s,y)\in[0,t]\times(-\infty,-X_s]}1_{\{\tau>s\}}1_{\{X_s>0\}}\;m\times v(ds\times dy).
\end{align*}
Let $B$ be the first integral in the second equality, then $B\leq m([0,t]\cap\{s;X_s=0\})\times v(-\infty,0]$. However, from Fubini-Tonelli Theorem, we have \begin{equation}\label{eq:leb}
     \expect{m([0,t]\cap\{s;X_s=0\})}=\expect{\int_0^t1_{\{X_s=0\}}\;ds}=\int_0^t\pr(X_s=0)\;ds=0,
\end{equation}
where the last equality is due to continuity of distribution of $X_s$.  Therefore, $m([0,t]\cap\{s;X_s=0\})$ and so $B$ are almost surely Zero\footnote{Note that in the case of $v(-\infty,0]=\infty$, the usual convention of measure theory is applied, i.e. $0\times\infty=0$.}. 

On the other hand, the process $X$ is quasi-left-continuous (see Lemma 3.2 of Kyprianou (2006)) which concludes that for every $t>0$, $\pr(X_t=X_{t^-})=1$, hence by a similar argument as above, the following equality holds almost surely 
$$A=\int1_{(s,y)\in[0,t]\times(-\infty,-X_s]}1_{\{\tau>s\}}1_{\{X_s>0\}}1_{\{X_{s^-}>0\}}\;m\times v(ds\times dy).$$

\textbf{Step 2.} Here we show that $\pr(X_{\tau^-}=0)=0$. Note that since $\tau$ is not predictable, quasi-left-continuity is not applicable. First, we have that $\pr(X_{\tau^-}=0)=\pr(X_{\tau^-}=0,\Delta X_\tau\neq0)+\pr(X_{\tau^-}=0,\Delta X_\tau=0)$, and $\pr(X_{\tau^-}=0,\Delta X_\tau=0)=\pr(X_{\tau^-}=0, X_\tau=0)\leq\pr( X_\tau=0)=0$, by Remark \ref{remark:tist}. Hence
\begin{align*}
\pr(X_{\tau^-}=0)&=\pr(X_{\tau^-}=0, \Delta X_\tau\neq0)\leq\pr\left(\sum_{0\leq s<\infty}1_{\{X_{s^-}=0\}}1_{\{\Delta X_s\neq0\}}\geq1\right)\\&\leq\expect{\sum_{0\leq s<\infty}1_{\{X_{s^-}=0\}}1_{\{\Delta X_s\neq0\}}}\\
&=\expect{\int_0^\infty\int_\R\phi(s,x)\;J_X(ds\times dx)}\\
&=\expect{\int_0^\infty\int_\R\phi(s,x)\;v(dx)\times ds}\\
&=\expect{m(s\in[0,\infty);X_{s^-}=0)\times v(\R-{0})},
\end{align*}
where $\phi(s,x)=1_{\{X_{s^-}=0\}}1_{\{x\neq0\}}$ is predictable, and so the compensation formula is applicable. By a similar argument to Step 1, we deduce that $m(s\in[0,\infty);X_{s^-}=0)=0$, almost surely. Therefore $\pr(X_{\tau^-}=0)=0.$

\textbf{Step 3.} From Step 1, we almost surely have
$$A\leq\int1_{\{(s,y)\in[0,t]\times(-\infty,-\inf_{0\leq s<t\wedge\tau}X_s]\}}1_{\{\tau>s\}}1_{\{X_s>0\}}1_{\{X_{s^-}>0\}}\;m\times v(ds\times dy).$$ 
Using Step 2 and the fact that the  process $X$ is c\`adl\`ag, one can show that $\alpha(t,\tau)=\inf_{0\leq s<t\wedge\tau}X_s>0$ almost surely, hence
\begin{align*}
      A &\leq \int 1_{\{s\in[0,t]\}}1_{\{y\in(-\infty,-\alpha(t,\tau)]\}}1_{\{\alpha(t,\tau)>0\}}\;m\times v(ds\times dv)\\
         &=tv\left(-\infty,-\alpha(t,\tau)\right)1_{\{\alpha(t,\tau)>0\}}.
\end{align*}
Because $v$ is a Radon measure this shows that $A<\infty$, almost surely.
\end{proof}

\end{document}